\documentclass[11pt]{article}

\usepackage[utf8]{inputenc}
\usepackage[T1]{fontenc}

\usepackage{fullpage}

\usepackage[small]{caption}
\usepackage{subcaption}

\usepackage{amsmath}
\usepackage{amsfonts}
\usepackage{amssymb}
\usepackage{amsthm}
\usepackage{amstext}
\usepackage{thmtools}

\usepackage{csquotes}
\usepackage{xcolor}
\usepackage{graphicx}

\usepackage{mathtools}
\usepackage{xspace}
\usepackage{enumerate}
\usepackage{algorithmicx}
\usepackage{algorithm}
\usepackage{etoolbox}
\usepackage[noend]{algpseudocode}
\usepackage{paralist}

\usepackage{url}
\usepackage{hyperref}
\usepackage[capitalize]{cleveref}

\usepackage{libertine}
\usepackage[libertine]{newtxmath}
\usepackage{microtype}

\newtheorem{theorem}{Theorem}
\newtheorem{corollary}[theorem]{Corollary}
\newtheorem{lemma}[theorem]{Lemma}

\newtheorem{proposition}[theorem]{Proposition}

\theoremstyle{definition}
\newtheorem{definition}[theorem]{Definition}

\renewcommand{\S}{\mathcal{S}}

\newcommand{\D}{\mathcal{D}}
\newcommand{\Hyp}{\mathcal{H}}
\newcommand{\E}{\mathbb{E}}

\DeclarePairedDelimiter{\abs}{\lvert}{\rvert}
\DeclareMathOperator*{\argmin}{arg\,min}
\newcommand{\bigO}{\mathcal{O}} 
\newcommand{\A}{\ensuremath{\mathcal{A}}\xspace}

\newcommand{\B}{\ensuremath{\mathcal{B}}\xspace}

\newcommand{\alg}{{\normalfont\textsc{Alg}}}
\newcommand{\opt}{{\normalfont\textsc{Opt}}}

\newcommand{\cS}{\ensuremath{\mathcal{S}}}

\newcommand{\pred}{\hat{\sigma}}
\newcommand{\perfectpred}{\sigma}

\newcommand{\dense}{\mu}

\newcommand{\dualSa}{\hat{a}}
\newcommand{\dualSb}{\hat{b}}
\newcommand{\dualVa}{a}
\newcommand{\dualVb}{b}

\newcommand{\errR}{\eta^R}
\newcommand{\errS}{\eta^S}

\newcommand{\speed}{s}

\newcommand{\abc}[3]{$\textup{#1}|#2|#3$}

\definecolor{bluegreen}{HTML}{33615E}
\definecolor{lime}{HTML}{56A778}
\definecolor{yellow}{HTML}{E5CB6D}
\definecolor{orange}{HTML}{FB8318}
\definecolor{lightred}{HTML}{DD6143}

\definecolor{proved}{HTML}{85A24A}
\definecolor{expected}{HTML}{F4AB36}
\definecolor{close}{HTML}{117F7D}
\definecolor{unknown}{HTML}{C1181E}

\usepackage[colorinlistoftodos,prependcaption,textsize=scriptsize]{todonotes}

\title{Permutation Predictions for Non-Clairvoyant Scheduling}

\author{
Alexander Lindermayr\thanks{Faculty of Mathematics and Computer Science, University of Bremen, Germany. \emph{\{linderal,nmegow\}@uni-bremen.de}} 
\and
Nicole Megow\footnotemark[1] 
}

\date{}

\begin{document}

\maketitle

\begin{abstract}
  In non-clairvoyant scheduling, the task is to find an online strategy for scheduling jobs with a priori {\em unknown} processing requirements with the objective to minimize the total (weighted) completion time. We revisit this well-studied problem in a recently popular learning-augmented setting that integrates (untrusted) predictions in online algorithm design. 
  While previous works used predictions on processing requirements, 
  we propose a new prediction model, which provides a relative order of jobs which could be seen as predicting algorithmic actions rather than parts of the unknown input. We show that these predictions have desired properties, admit a natural error measure as well as algorithms with strong performance guarantees and that they are learnable in both, theory and practice. We generalize the algorithmic framework proposed in the seminal paper by Kumar et al.~(NeurIPS'18) and present the first learning-augmented scheduling results for weighted jobs and unrelated machines. We demonstrate
  in empirical experiments the practicability and superior performance 
  compared to the previously suggested single-machine algorithms. 
\end{abstract}

\thispagestyle{empty}

\section{Introduction}

Non-clairvoyant scheduling requires to schedule jobs without knowing their processing requirements a priori. This is a fundamental problem and has been studied extensively in many variations~\cite{DBLP:journals/tcs/MotwaniPT94,DBLP:journals/ipl/KimC03a,DBLP:conf/ipps/BeaumontBEM12, DBLP:conf/focs/ImKMP14,DBLP:journals/jacm/ImKM18}.

We consider non-clairvoyant scheduling with the objective of minimizing the sum of weighted completion times in different settings. Generally, we are given a set of jobs, each job with individual weight and unknown processing time, possibly arriving online at its release date. %
All jobs must be scheduled on a single or identical parallel machines; preemption is allowed.  Using classical scheduling notation, we refer to the problems we consider as the non-clairvoyant versions of \abc{1}{pmtn}{\sum w_jC_j} and \abc{P}{r_j,pmtn}{\sum w_jC_j}. We also investigate %
	non-clairvoyant scheduling on unrelated machines, denoted by \abc{R}{r_j, pmtn}{\sum w_jC_j}, where jobs may have very different processing times on each of the machines, but a machine-dependent processing rate is given. (Precise definitions follow later.) 

The performance of online algorithms is typically assessed by \emph{competitive analysis}. An online algorithm is~$\rho$-{\em competitive} if, for all instances~$I$, the algorithm has cost~$\alg(I) \leq \rho \cdot \opt(I)$, where~$\opt(I)$ is the objective value of an optimal solution for~$I$. 

Non-clairvoyant algorithms assign processing \emph{rates} to jobs and assume time sharing, that is, parallel processing of jobs with rates that sum up to at most one per machine and per job. One could see this as processing each job by a certain amount in every infinitesimal time interval. The most prominent strategy is the Round-Robin~(RR) algorithm, which assigns equal rates to all alive jobs and is~$2$-competitive for \abc{1}{pmtn}{\sum C_j}, which is best possible~\cite{DBLP:journals/tcs/MotwaniPT94}. The same guarantee is possible using a natural generalization of RR to weighted jobs~\cite{DBLP:journals/ipl/KimC03a} and/or to identical machines~\cite{DBLP:conf/ipps/BeaumontBEM12,DBLP:journals/tcs/MotwaniPT94}.
Scheduling on unrelated machines is much harder and requires careful migration between machines~\cite{DBLP:conf/soda/GuptaIKMP12}. Nevertheless, it is possible to compute rates proportional to job properties and machine constraints and obtain an~$\bigO(1)$-competitive algorithm~(\cite{DBLP:journals/jacm/ImKM18}; also implicitly in~\cite{DBLP:conf/focs/ImKMP14}). %

The assumption of non-clairvoyance seems  too strong for many applications. While the exact processing time might be unknown, often some estimate is available, e.g., 
extracted information from past data is commonly used to predict the future.
The recently emerging line of research on \emph{learning-augmented} algorithms proposes 
to design algorithms that have access to additional (possibly erroneous) input, called \emph{prediction}, to achieve an improved performance if the prediction is accurate while performing not much worse than algorithms without access to predictions, if the predictions are completely wrong. Ideally, the performance of a learning-augmented algorithm is a function of the quality of the prediction for some well defined error measure. Here, defining an appropriate error measure is a key task. 
Given a definition for the prediction error~$\eta \geq 0$ that quantifies the quality of the prediction, %
the goal is express the competitive ratio of the algorithm by a monotone function~$f(\eta)$. %
A learning-augmented algorithm is called~$f(0)$-consistent (in case of perfect prediction) and~$\beta$-robust if~$f(\eta) \leq \beta$ for all possible errors~$\eta$. 

Recent work on non-clairvoyant scheduling with predictions~\cite{DBLP:conf/nips/PurohitSK18,DBLP:conf/nips/WeiZ20,DBLP:conf/spaa/Im0QP21} studies the single-machine problem \abc{1}{pmtn}{\sum C_j} with predicted processing requirements~${\{y_j\}}_{j \in J}$, which we call \emph{length predictions}. 
Commonly, we distinguish two categories of prediction models: either predict parts of the online input (\emph{input-predictions})~\cite{DBLP:conf/icml/LykourisV18,DBLP:conf/nips/PurohitSK18,DBLP:conf/nips/BamasMRS20,AzarPT22,AzarLT22} or algorithmic actions (\emph{action-predictions})~\cite{DBLP:conf/icml/AntoniadisCE0S20,DBLP:conf/soda/LattanziLMV20,DBLP:conf/nips/BamasMS20}. Length predictions clearly fall in to the first category.

In their seminal paper~\cite{DBLP:conf/nips/PurohitSK18}, Kumar et al.~propose an algorithm that is controlled by a parameter~$\lambda \in (0,1)$, which can be seen as an indicator of the algorithm's trust in the accuracy of the prediction. Measuring the quality of a prediction ${\{y_j\}}_{j \in J}$ w.r.t.\ the actual processing requirements ${\{p_j\}}_{j \in J}$ by the~$\ell_1$ metric~($\ell_1 = \sum_{j \in J} \abs{p_j - y_j}$), they prove a competitive ratio of at most~$(1/(1 - \lambda))(1 + n \ell_1 / \opt)$ while also maintaining a robustness factor of~$2 / \lambda$. 
However, the~$\ell_1$-metric does not seem to distinguish well between ``good'' and ``bad'' predictions, as has been noted recently by Im et al.~\cite{DBLP:conf/spaa/Im0QP21}. 
They argue that, intuitively, the linear error measure $\ell_1$ is incompatible with the sum of weighted completion time objective 
and using $n \cdot \ell_1$ as upper bound %
may overestimate the ``actual'' error, substantially.

Im et al.~\cite{DBLP:conf/spaa/Im0QP21} propose a different error measure $\nu$ that satisfies certain  desired properties and is based on the optimal solution of artificial instances mixing $y_j$ and $p_j$. It satisfies~$\ell_1 \leq \nu \leq n \ell_1$. Using this error, they design a learning-augmented randomized algorithm with competitive ratio $\min\{1+\lambda+{\bigO(1/\lambda^3 \log (1/\lambda))\cdot \nu}/{\opt}, {2}/{\lambda}\}$, for sufficiently small $\lambda>0$, in expectation. %
Their algorithm is quite sophisticated, requires large constants to diverge from RR (we give more details later), and it seems very challenging to generalize it to scheduling settings with release dates, weights or even heterogeneous machines. Further, the error measure $\nu$ is still sensitive to changes in the predicted job lengths which would not affect an optimal schedule at all which seems an undesired property.

\paragraph{Our contribution} In this work, we contribute to non-clairvoyant scheduling with predictions in two ways: (i) we propose a new prediction model with a new error definition, as an alternative to length predictions studied so far, 
and (ii) we revisit the classical idea of time sharing and develop a general framework for designing learning-augmented scheduling algorithms for more general settings, beyond the simple single-machine setting.

We propose a novel prediction model for scheduling problems, which we call \emph{permutation prediction model}. Intuitively, it provides    
a permutation of jobs suggesting %
a priority order for scheduling. In a way, this is an action-prediction in contrast to previously studied input-predictions. 
The idea is that, instead of predicting job lengths, we take structural properties of an input instance into account that an optimal algorithm may exploit.  
Notice that for minimizing the sum of weighted completion time, the {\em Weighted Shortest Remaining Processing Time (WSPT)} order, i.e., jobs in order of weight over processing time ratios, %
 has proven to be useful in various settings. Indeed, for the non-clairvoyant version of \abc{1}{pmtn}{\sum w_jC_j}, knowing the WSPT order of jobs would be sufficient to determine an optimal schedule~\cite{smith1956various}. While this knowledge is not sufficient for optimally scheduling with release dates and/or on multiple machines, it still admits strategies with good approximations on an optimal solution~\cite{DBLP:journals/orl/MegowS04,DBLP:conf/soda/AnandGK12,DBLP:journals/mor/GuptaMUX20}. For unrelated machines, we also include a job-to-machine assignment in the prediction model.  

Clearly, a WSPT-based permutation prediction could be derived from a length prediction. The advantage of our model is that it is much more compact, %
captures a crucial structural property of an optimal solution and makes error measures less vulnerable to small noise in the prediction compared to the length prediction model. 

As a key contribution, we define a new, meaningful error measure that quantifies the impact of an error in the prediction to an algorithm's cost explicitly in terms of the objective function. %
It has several desirable properties such as 
\begin{inparaenum}[$(i)$]
\item monotonicity and
\item Lipschitzness (both highly advertised recently by Im el al.~\cite{DBLP:conf/spaa/Im0QP21}),
\item theoretical learnability of our prediction model with respect to the error definition, which we show by proving that our predictions are efficiently PAC-learnable in the agnostic sense, as well as
\item practical learnability, which we demonstrate in empirical experiments, showing that our implemented learning algorithm quickly improves the performance of our scheduling algorithms and appears superior to previously presented algorithms.
\end{inparaenum}

Further, we revisit the algorithmic technique of time sharing introduced by Kumar~et~al.~\cite{DBLP:conf/nips/PurohitSK18} 
in their seminal work on  non-clairvoyant scheduling with predictions. We extend this technique  to a general framework for designing learning-augmented scheduling algorithms allowing for release dates, job weights and unrelated machines. 
As a main contribution, we give the first algorithm for non-clairvoyant scheduling with predictions on unrelated machines and prove strong performance bounds, smoothly degrading with prediction quality. 
More precisely, we show for the permutation prediction model and two appropriate error definitions $\errS$ and $\errR$ that there exists for every~$\lambda \in (0,1)$ a learning-augmented non-clairvoyant online algorithm for minimizing the total weighted completion time on
  \begin{enumerate}[(i)]
    \item a single machine, \abc{1}{pmtn}{\sum w_jC_j}, with a competitive ratio of at most 
    \[ 
      \min\left\{ \frac{1}{1-\lambda} \left(1 + \frac{\errS}{\opt} \right), \frac{2}{\lambda} \right\},
    \]
    \item $m$~identical machines with release dates, \abc{P}{r_j,pmtn}{\sum w_jC_j}, with a competitive ratio of at most 
    \[\min\left\{ \frac{1}{1-\lambda} \left(2 + \frac{\errS}{m \cdot \opt} \right), \frac{3}{\lambda} \right\}, \text{ and}
    \]
    \item unrelated machines with release dates, \abc{R}{r_j,pmtn}{\sum w_jC_j}, with a competitive ratio of at most 
    \[ 
      \min\left\{ \frac{1}{1-\lambda} \left(5.8284 + \frac{\errR}{\opt} \right), \frac{128}{\lambda} \right\}.
    \]
  \end{enumerate}

Our framework requires a clairvoyant and a non-clairvoyant algorithm for a given scheduling problem, both of them must satisfy a certain monotonicity property. Then, we design a learning-augmented variation of the clairvoyant algorithm that admits a competitive ratio as a function of the error. %
Intuitively, the errors~$\errR$ and $\errS$ measure how much an erroneous prediction influences the objective value compared to an accurate %
prediction. For a single and identical machines, we require even less predicted information (no machine assignment) and the simpler measure~$\errS$ suffices.

While we use non-clairvoyant algorithms as a black box from the literature, the new contribution lies in proving error-dependent competitive ratios for monotone clairvoyant algorithms that use predictions as input. This may require designing new algorithms. In particular, we show a competitive ratio of $3 + 2 \sqrt{2} \approx 5.8284$ for a natural Greedy algorithm for the clairvoyant %
problem \abc{R}{r_j,pmtn}{\sum w_jC_j}. This does not match the recent and best known deterministic bound of $3$~\cite{DBLP:conf/icalp/BienkowskiKL21}, but our algorithm satisfies the desired properties of being error-sensitive and monotone.

\paragraph{Further related work}
There has been significant interest in the recent framework of learning-augmented online algorithms. Many problems have been considered, e.g., caching~\cite{DBLP:conf/soda/Rohatgi20,DBLP:conf/approx/Wei20,DBLP:conf/icml/AntoniadisCE0S20}, further scheduling~\cite{DBLP:conf/soda/LattanziLMV20,DBLP:conf/stoc/AzarLT21,AzarLT22,DBLP:conf/nips/WeiZ20,DBLP:conf/nips/BamasMRS20,DBLP:conf/innovations/Mitzenmacher20,DBLP:conf/acda/Mitzenmacher21,DBLP:conf/innovations/ScullyGM22},
rent-or-buy problems~\cite{DBLP:conf/nips/PurohitSK18,DBLP:conf/icml/GollapudiP19, DBLP:conf/icml/AnandGP20,DBLP:conf/nips/Banerjee20,DBLP:conf/nips/WangLW20,AntoniadisCEPS21,DBLP:conf/nips/WeiZ20}, 
paging~\cite{DBLP:conf/icalp/JiangP020,DBLP:conf/innovations/EmekKS21,DBLP:conf/soda/BansalCKPV22},
graph problems~\cite{DBLP:journals/corr/EberleLMNS22,AzarPT22,DBLP:journals/corr/XuM22,DBLP:conf/innovations/LindermayrMS22},
secretary problems~\cite{DBLP:conf/nips/AntoniadisGKK20,DBLP:conf/sigecom/DuttingLLV21},
matching~\cite{DBLP:conf/icml/AntoniadisCE0S20,DBLP:conf/esa/LavastidaM0X21,DBLP:conf/acda/LavastidaM0X21} and many more.

Non-clairvoyant and clairvoyant online scheduling models have been studied extensively; see the surveys~\cite{pruhsST04,sgall98}. Most relevant for our work are WSPT-based algorithms such as~\cite{DBLP:journals/orl/MegowS04,DBLP:conf/soda/AnandGK12,DBLP:journals/mor/GuptaMUX20}.

\paragraph{Paper organization}
In~\Cref{sec:prediction-model} we give precise definitions for the problem, prediction model and error measure. Then, we introduce our algorithmic framework and apply it to concrete scheduling problems in~\Cref{sec:framework}.
We prove efficient PAC learnability of our predictions in~\Cref{sec:learnability} and discuss empirical results in~\Cref{sec:experiments}.

\section{Problem and prediction model}%
\label{sec:prediction-model}

\subsection{Problem definition}
We consider the problem of scheduling~$n$ jobs~$J = \{1,\ldots,n\}=:[n]$ preemptively on~$m$ unrelated machines. Every job~$j \in J$ has an associated weight~$w_j$ and processing requirement~$p_{j}$. Further, for every machine~$i \in [m]$ there is given a rate~$\ell_{ij}$ which is the amount of processing that job~$j$ receives if it is processed one time unit on machine~$i$, resulting in a total processing time $p_{ij} = \ell_{ij} \cdot p_{j}$ if job $j$ is scheduled on machine~$i$. Jobs arrive online at their individual release dates~${\{r_j\}}_{j \in J}$. A non-clairvoyant online algorithm has to schedule jobs $J$ on the given machines, but is oblivious to unreleased jobs and has no information on  processing requirements.
The objective is to minimize the weighted sum of completion times~$\sum_{j \in J} w_j C_j$, where the completion time~$C_j$ of a job~$j$ is the first point in time when it has been processed for~$p_j$ units. In the standard three-field notation, this problem is denoted as non-clairvoyant version of \abc{R}{r_j,pmtn}{\sum w_jC_j}. Note that a non-clairvoyant algorithm is oblivious to the processing requirement~$p_j$ but needs access to the machine rates~$\ell_{ij}$ to admit a constant competitive ratio~\cite{DBLP:conf/focs/ImKMP14}.
The setting where $\ell_{ij} = 1$ for all jobs~$j$ and machines~$i$ is called identical machine setting, \abc{P}{r_j,pmtn}{\sum w_jC_j}, and the single machine setting without release dates~is~\abc{1}{pmtn}{\sum w_jC_j}.

\subsection{Permutation prediction model}

We propose a prediction model that is heavily inspired by the relevance of the WSPT order (jobs ordered by non-increasing densities~$\dense_{ij} = w_j / p_{ij}$) for scheduling to minimize the total weighted completion time. %

For a %
scheduling instance with job set $[n]$ and a single or multiple identical machines, our prediction is a permutation $\pred: [n]\rightarrow [n]$ of all jobs. Given the aforementioned power of the WSPT order, we %
call the associated permutation of jobs, $\perfectpred$, \emph{perfect prediction}. 

On unrelated machines, the job-to-machine assignment crucially matters. Therefore, we add such an assignment to our prediction. In this most general model, our prediction is defined as~$\pred = {\{\pred_i\}}_{i \in [m]}$, where~$\pred_i$ is the permutation of jobs assigned to machine~$i$, and every job is assigned to exactly one machine. We denote the machine to which job~$j$ is assigned in~$\pred$ by~$m(\pred,j)$. Given a scheduling instance without release dates and the optimal job-to-machine allocation,  it would be optimal to schedule jobs in WSPT order on each machine individually. Therefore, we speak of perfect prediction $\perfectpred={\{\perfectpred_i\}}_{i \in [m]}$, if $\perfectpred_i$ involves exactly those jobs that are scheduled in an optimal solution on machine $i$ and orders them in WSPT order, for each $i\in [m]$.

In the permutation prediction model, jobs still arrive online and, at any time,  an algorithm has access only to predictions on jobs that have been released already. At any release date, the permutation is updated consistently with the previous permutation. That is, the prediction model is not allowed to change the relative order of previously known jobs.

\subsection{Prediction error}

The prediction error defines a measure for the quality of a prediction. It is a crucial element in the design of learning-augmented algorithms. %
Intuitively, the error measure shall quantify the impact that an erroneous prediction has on an (optimal) scheduling algorithm. It is not unnatural to express the error as~$|\opt(\pred)-\opt(\perfectpred)|$, as has been done %
in~\cite{DBLP:conf/nips/BamasMS20,DBLP:journals/corr/EberleLMNS22,DBLP:conf/innovations/LindermayrMS22}, but for more complex scheduling environments the optimal solution is hard to compute and, more importantly,  this error could be even negligible whereas the impact of running an optimal algorithm with the wrong prediction could be significant. The latter is what we want to quantify.
 
In more detail, our error measure shall capture the change in the cost that an optimal schedule must face when two jobs~$j$ and~$j'$ are inverted in a prediction~$\pred$ with respect to~$\perfectpred$. For example, on a single machine without release dates, if~$j$ and its successor~$j'$ in~$\pred$ are swapped in~$\perfectpred$, the schedule that follows~$\pred$ pays an additional cost of~$w_{j'} p_j$ but saves~$w_j p_{j'}$ compared to the schedule that follows~$\perfectpred$. However, in presence of release dates and on multiple machines, just knowing the orders may not allow us to express the change in the exact optimal cost. Therefore, we rely on an approximation as a surrogate for the optimal cost, namely, the {\em change in the sum of weighted completion times when preemptively scheduling jobs in the given priority order},~$\pred$ resp. $\perfectpred$. %

We define two different error measures.  
Firstly, we define our simple error~$\errS$ for predictions that consist of a single permutation on all jobs.
Then, our general measure~$\errR$ describes the quality of permutation predictions with predicted job assignments, $\perfectpred={\{\perfectpred_i\}}_{i \in [m]}$. We show that $\errS$ is special case of~$\errR$.

\begin{definition}%
\label{def:errS}
	For an instance of non-clairvoyant scheduling %
	with permutation prediction $\pred$ consisting of a single permutation, and WSPT order $\perfectpred$, let $\mathcal{I}(J,\pred)=\{(j',j) \in J^2 \mid \perfectpred(j') < \perfectpred(j) \land \pred(j') > \pred(j) \}$ be the set of inverted job pairs. The prediction error of $\pred$ is defined as 
	 \[
     \errS(J,\pred) = \sum_{(j',j) \in \mathcal{I}(J,\pred)} (w_{j'}p_{j} - w_j p_{j'}).
    \]
\end{definition}

This error measures the \emph{exact} change in the objective value, in the absence of release dates. The single permutation prediction and this error will be sufficient for designing algorithms with appealing error-dependency for a single and parallel identical machines. 

For scheduling on unrelated machines and predictions including a job assignment, we need a more elaborate error definition which, nevertheless, follows the same idea. 
Given an instance with job set~$J$ and prediction $\pred={\{\pred_i\}}_{i \in [m]}$, 
we define for every job~$j \in J$ a partial error~$\eta_j$, which measures how much the different positions of~$j$ 
in~$\pred$ resp.~$\perfectpred$ increases the objective value assuming preemptive scheduling according to a given permutation.

To this end, we consider for an arbitrary assignment and permutation $\pi={\{\pi_i\}}_{i\in[m]}$ the schedule that processes at every point in time~$t$ on every machine~$i$ the available job~$j'$ with~$m(\pi,j') = i$ that has the highest priority in~$\pi_i$. 
Let~$A(j)$ denote the set of jobs that are released but unfinished by time $r_j$ and that are assigned to machine~$m(\pi,j)$. Note that~$j \in A(j)$. For two jobs~$j$ and~$j'$ with~$r_j = r_{j'}$ and~$m(\pi, j) = m(\pi, j')$, we assume that they are assigned to the machine in order of their indices.
By denoting the remaining processing requirement of job~$j$ at time~$t$ by~$p_j(t)$, and~$p_{ij}(t) = \ell_{ij}p_j(t)$, the
increase in the schedule's objective value for adding job $j$ to machine~$i = m(\pi, j)$ is equal to
\[
    W_{j}(J, \pi) =  p_{ij} \sum_{\substack{j' \in A(j) \\ \pi_i(j') > \pi_i(j)}} w_{j'} + w_j \Bigg( r_j + \sum_{\substack{j' \in A(j) \\ \pi_i(j') \leq \pi_i(j)}} p_{ij'}(r_j) \Bigg).
\]

\begin{definition}
	For an instance of non-clairvoyant scheduling with permutation prediction $\pred={\{\pred_i\}}_{i \in [m]}$ and perfect prediction $\perfectpred={\{\perfectpred_i\}}_{i \in [m]}$, the prediction error for job $j\in J$ is defined as 
		\[ \eta_j(J,\pred) = W_{j}(J, \pred) - W_j(J, \perfectpred). \]
	The prediction error of $\pred$ is given by~$\errR(J,\pred) = \sum_{j\in J} \eta_j(J,\pred)$.
\end{definition}

It is not difficult to see that $\errR$ reduces to the compact error measure $\errS$ for predictions that consist of a single permutation (without machine assignment) and without release dates.
\begin{proposition}\label{lemma:inversions}
	For a job set $J$ and a permutation prediction $\pred$, if $\pred$ is a single permutation and $r_j=0$, for all $j\in J$, then 
  	\[
      \errR(J,\pred) = \errS(J,\pred). %
    \]
\end{proposition}
\begin{proof}
  Let~$j \in J$. Observe that $\eta_j(J,\pred) =  W_{j}(J,\pred) - W_j(J,\perfectpred)$ %
  equals under the stated assumptions
  \begin{equation*} \sum_{\substack{j' \in J\\\pred(j) < \pred(j')}} w_{j'} p_{j}  + \sum_{\substack{j' \in J\\ \pred(j') < \pred(j)}} w_j  p_{j'} - \sum_{\substack{j' \in J \\\perfectpred(j) < \perfectpred(j')}} w_{j'} p_{j} -  \sum_{\substack{j' \in J \\ \perfectpred(j') < \perfectpred(j)}} w_j p_{j'}.
  \end{equation*}

  Combining the first with the third sum and the second with the fourth gives
  \begin{align*}
    \sum_{\substack{\substack{j' \in J\\\perfectpred(j) > \perfectpred(j')}\\ \pred(j) < \pred(j')}} w_{j'} p_j  
    - \sum_{\substack{\substack{j' \in J\\\perfectpred(j') < \perfectpred(j)}\\ \pred(j') > \pred(j)}} w_{j} p_{j'}
    = \sum_{\substack{\substack{j' \in J\\\perfectpred(j) > \perfectpred(j')}\\ \pred(j) < \pred(j')}} (w_{j'} p_j - w_{j} p_{j'}). 
  \end{align*}

  Summing over all jobs and inversion pairs $\mathcal{I}$ yields %
  \[ \sum_{(j',j) \in \mathcal{I}(J,\pred)} (w_{j'}p_{j} - w_j p_{j'}) = \errS(J,\pred). \qedhere\]
\end{proof}

\subsection{Properties of the error measure}

Our new error measure satisfies several desired properties such as 
\begin{inparaenum}[$(i)$]
  \item monotonicity,
  \item Lipschitzness,
  \item theoretical learnability, and
  \item practical learnability.
\end{inparaenum}

Im et al.~\cite{DBLP:conf/spaa/Im0QP21} advocate particularly the first two properties. 
\emph{Monotonicity} requires, in the length prediction model, that the error grows 
as more length predictions become incorrect. 
In our setting, we have~$\eta(\pred) = 0$ if~$\pred = \perfectpred$, 
and for any inversion added to~$\pred$, the error grows. 
This is because an inversion~$(j',j) \in \mathcal{I}$ increases the error by $w_{j'}p_{j} - w_j p_{j'}$, since~$\perfectpred(j') < \perfectpred(j)$ implies~$w_{j'} / p_{j'} \geq w_j / p_j$. 
Thus, our definition satisfies monotonicity. 
 
\emph{Lipschitzness} requires the error to bound the absolute difference of the optimal objective values 
for the actual and predicted instance from above. %
Our error definition {\em precisely} measures the cost between a solution 
that follows~$\pred$ and one that follows~$\perfectpred$, when scheduling the actual instance preemptively according to the given order.
Hence, our error measures immediately satisfy Lipschitzness for our prediction~setup. 

Our prediction model is \emph{theoretically learnable} in the framework of PAC-learnability~\cite{shalevB14ML}. 
We show that permutations %
are efficiently PAC-learnable in the agnostic sense w.r.t.\ our error definition~(\Cref{sec:learnability}).
While this theoretic result gives a rather large bound on the required number of samples to get a low prediction error, we further demonstrate that our 
predictions are \emph{learnable and useful in practice}. We 
implement a learning algorithm and show that even a small number of 
seen samples results in a drastic performance improvement of our 
algorithm in practical instances~(\Cref{sec:experiments}). 

In general, it is difficult to compare different prediction and error models. However, we can convert a given length prediction into a permutation prediction by simply computing the WSPT order based on the predicted processing requirements. %
For the case of unrelated machines, we further require predicted 
machine assignments. %
This conversion allows us to compare our error to the previously proposed measures~$\nu$ and~$\ell_1$ for the case of %
\abc{1}{pmtn}{\sum C_j}.

Firstly, we note that our error~$\errS$ is less vulnerable than~$\nu$ and~$\ell_1$ to changes in the predicted instance which do not affect the \emph{structure} of an optimal solution.
Indeed, the optimal solution of an instance with~$p_j=j$ for all~$j \in [n]$ has the same structure as the optimal solution of a predicted instance with~$y_j = j - 1$ for all~$j \in [n]$. One would expect a small error, and indeed~$\errS = 0$. In contrast, previously defined errors are large:~$\nu = \opt(\{\max\{p_j,y_j\}\}) - \opt(\{\min\{p_j,y_j\}\}) = n(n+1)/2 - (n-1)n/2 = n$ and~$\ell_1 = \sum_{j \in J} \abs{p_j - y_j} =n$. 
This shows that our prediction and error seem to capture well the relevant characteristics of an input-prediction in terms of derived actions, while~$\nu$ and~$\ell_1$ also track insignificant numerical differences between the actual and predicted instances. 

In contrast to this example, there are other instances where~$\nu$ and~$\ell_1$ underestimate the actual difficulty that is caused by the inaccuracy of the prediction given to an (optimal) algorithm. %
Im et al.~\cite{DBLP:conf/spaa/Im0QP21} give such an example with~$p_1=y_1=\ldots=p_{n-1}=y_{n-1}=1$ and~$p_n={n^2}$ but~$y_n = 0$. While the structural difference of the optimal solutions for predicted and true values is large ($\errS=\Omega(n^3)$) %
the other error definitions only measure~$\nu = n^2 + n$ and~$\ell_1 = n^2$.

It is not difficult to see that our prediction error never exceeds~$n \ell_1$.
\begin{proposition}
  For any instance of~\abc{1}{pmtn}{\sum C_j} and length prediction,~$\errS \leq n \cdot \ell_1$.
\end{proposition}
\begin{proof}
  Consider an instance with job set $J$ and length prediction~${\{y_j\}}_{j \in [n]}$. Let~$\pred$ be the corresponding predicted permutation.
 Since~$(j',j) \in \mathcal{I}(J, \pred)$ implies~$\pred(j') > \pred(j)$, which must be due to~$y_{j} \leq y_{j'}$, we conclude
  \[
    \errS(J, \pred) 
    = \sum_{(j',j) \in \mathcal{I}(J, \pred)} p_j - y_j + y_j - y_{j'} + y_{j'} - p_{j'} 
    \leq \sum_{(j',j) \in \mathcal{I}(J, \pred)} \abs{p_j - y_j} + \abs{p_{j'} - y_{j'}} \leq n \ell_1. \qedhere
  \]
\end{proof}

Our results for non-uniform job weights on a single and identical machines translate to the length prediction model, as one can similarly show that $\errS$ is bounded by the natural weighted generalization of~$n \cdot \ell_1$, that is $\sum_{j' \in J} w_{j'} \sum_{j \in J} \abs{p_j - y_j}$.

\section{Preferential Time Sharing}\label{sec:framework}

We describe a framework for designing %
algorithms for non-clairvoyant scheduling with untrusted predictions, which we apply to several concrete scheduling settings in the following subsections. %

In their seminal paper, %
Kumar~et~al.~\cite{DBLP:conf/nips/PurohitSK18} proposed a single-machine 
time sharing algorithm for executing two algorithms `in parallel', a clairvoyant (assuming predicted processing times to be correct) and a non-clairvoyant algorithm. The rate, at which each of these algorithms is executed, is determined by the confidence parameter~$\lambda \in (0,1)$. 
We extend this idea to a general framework for scheduling jobs with non-uniform weights and arbitrary release dates on unrelated machines.

This technique requires that both algorithms are \emph{monotone}~\cite{DBLP:conf/nips/PurohitSK18}.
\begin{definition}
  A scheduling algorithm is \emph{monotone}, if for two 
  instances with identical inputs but actual job processing 
  requirements~$\{p_1,\ldots,p_n\}$ and~$\{p'_1,\ldots,p'_n\}$ such 
  that~$p_j \leq p'_j$ for all~$j \in [n]$, the objective value of the 
  algorithm for the first instance is at most its objective value for the 
  second one. 
\end{definition}

Given two monotone algorithms~$\A$ and~$\B$ and a confidence parameter~$\lambda \in (0,1)$, we define a new preemptive algorithm: we run on all machines and for every infinitesimal time interval, algorithm~\A in the first~$(1 - \lambda)$-fraction of the interval and algorithm~\B in the remaining~$\lambda$-fraction of the interval.
The new algorithm hides arrived jobs until they are released in the simulated, i.e., slowed down, schedule of~\A resp.~\B.
The following result generalizes a single-machine version without weights and release dates~\cite{DBLP:conf/nips/PurohitSK18}.

\begin{lemma}\label{lemma:time-slice-general}
	Given a parameter~$\lambda \in (0,1)$ and two monotonic algorithms with competitive ratios~$\rho_\A$ and~$\rho_\B$ for the online problem~\abc{R}{r_j,pmtn}{\sum_j w_j C_j}, there exists an algorithm 
  for the same problem with a competitive ratio~$\min \left\{ \frac  {\rho_\A}{1 - \lambda}, \frac{\rho_\B}{\lambda} \right\}$.
\end{lemma}
\begin{proof}
  Assume that the competitive ratios of \A and \B are at most~$\rho_\A$ and~$\rho_\B$. 
  By monotonicity of both algorithms, 
  whenever one algorithm processes a job, the other one will not have a higher objective value due to shorter processing requirements.
  Since we execute \A for a~$(1 - \lambda)$-fraction of time and \B for a~$\lambda$-fraction of time, the weighted completion time of a job increases by a factor of at most~$1/(1-\lambda)$ resp.~$1/\lambda$ compared to the schedules of~$\A$ resp.~$\B$, which implies the competitive ratio of the new algorithm.
\end{proof}

Our {\em Preferential Time Sharing} framework crucially builds on \Cref{lemma:time-slice-general} and takes as input two monotone algorithms, a clairvoyant algorithm~$\A^C$ with a competitive ratio of at most~$\rho_C$ and a non-clairvoyant algorithm~$\A^N$ with a competitive ratio of at most~$\rho_N$. Intuitively, the non-clairvoyant algorithm will ensure robustness, while the clairvoyant algorithm, being executed based on the given predictions, gives a good consistency. As~$\A^C$ will have access to predictions while being oblivious of true processing requirements, we call it {\em prediction-clairvoyant}. 
Our framework then gives, using~\Cref{lemma:time-slice-general} with $\A = \A^C$ and $\B = \A^N$, a time sharing algorithm with consistency~$\rho_C/(1-\lambda)$ and robustness~$\rho_N/\lambda$.

When aiming for error-sensitive guarantees, we require an error-dependent performance guarantee for~$\A^C$.
\begin{definition}  
  A prediction-clairvoyant algorithm is~$\eta$-\emph{error-dependent} for an error measure~$\eta$  if its objective value is bounded by~$\rho_C \cdot \opt(J) + \eta(J,\pred)$ for any instance~$J$ and prediction~$\pred$. 
\end{definition}

We note that these definitions are independent of the used prediction model. 
A straightforward consequence is as follows.

\begin{corollary}\label{coro:framework}
  Preferential Time Sharing with a monotone,~$\eta$-error-dependent algorithm~$\A^C$ with  competitive ratio at most~$\rho_C$ and a monotone, non-clairvoyant algorithm~$\A^N$ with competitive ratio at most~$\rho_N$ has, for every~$\lambda \in (0,1)$, a competitive ratio of~at~most
  \[
    \min \left\{ \frac{1}{1-\lambda} \left(\rho_C + \frac{\eta}{\opt} \right), \frac{\rho_N}{\lambda}\right\}
  \]
  for non-clairvoyant scheduling with predictions %
  \abc{R}{r_j, pmtn}{\sum w_j C_j}.
\end{corollary}

In the following subsections, we apply the Preferential Time Sharing framework %
to different concrete scheduling problems and prove our main algorithmic results. %
This requires: 
\begin{enumerate}[(i)]
  \item develop a monotone prediction-clairvoyant algorithm~$\A^C$ with error-dependent competitive ratio; and
  \item select an applicable non-clairvoyant monotone algorithm.
\end{enumerate}
By~\Cref{coro:framework}, both algorithms combined give the desired performance bounds for preemptive scheduling with predictions. While non-clairvoyant algorithms for our problems are available in the literature, our main contribution lies in designing prediction-clairvoyant algorithms with provable low error-dependency.

\subsection{Single machine}\label{sec:single-machine}
Consider non-clairvoyant scheduling of weighted jobs on a single machine, \abc{1}{pmtn}{\sum w_jC_j}.

\paragraph{Prediction-clairvoyant algorithm.} %
It is well-known that scheduling non-preemptively in the order given by $\pred$ gives the optimal schedule~\cite{smith1956various} if~$\pred$ coincides with the WSPT order. We refer to this algorithm as prediction-clairvoyant WSPT. It is monotone since, for a fixed prediction, shrinking a job does not affect~$\pred$ and only results in a lower completion time for this job and all its successors in~$\pred$. We now show that it is~$\errS$-error-dependent.

\begin{lemma}\label{lemma:single-machine-error}
  The prediction-clairvoyant WSPT algorithm is~$\errS$-error-dependent.
\end{lemma}

\begin{proof} 
	Consider an instance~$J$ with jobs being indexed by~$\perfectpred$, a prediction~$\pred$, and the schedule obtained by the prediction-clairvoyant WSPT algorithm.
  In this schedule, let~$d(j',j)$ denote the amount of job~$j'$ that has been processed before job~$j$ completed. Thus,~$d(j',j) = p_{j'}$ if and only if~$\pred(j') < \pred(j)$. This implies 
  \begin{align*}
  \alg(J,\pred)
      &= \sum_{j=1}^n w_j p_j + \sum_{j = 1}^n \sum_{j'=1}^{j-1} \left( w_j \cdot d(j',j) + w_{j'} \cdot d(j,j') \right) \\
      &= \sum_{j=1}^n w_j p_j +  \sum_{j = 1}^n \sum^{j-1}_{\substack{j'=1\\ \pred(j') < \pred(j)}} w_j p_{j'} +  \sum_{j = 1}^n \sum^{j-1}_{\substack{j'=1 \\ \pred(j') > \pred(j)}} w_{j'} p_j \\
      &= \sum_{j=1}^n w_j \sum_{j'=1}^{j} p_{j'} + \sum_{j = 1}^n \sum^{j-1}_{\substack{j'=1 \\ \pred(j') > \pred(j)}} (w_{j'} p_j - w_j p_{j'}) \\
      &= \opt(J) + \errS(J,\pred).
  \end{align*}

  The last equation holds since the first sum equals the objective value of the true WSPT schedule, i.e., a schedule according to~$\perfectpred$, which is optimal and the second sum equals~$\errS(J,\pred)$ by~\Cref{def:errS}, since we assumed the jobs to be indexed according~to~$\perfectpred$.
\end{proof}

\paragraph{Non-clairvoyant algorithm.} The {\em Weighted Round Robin (WRR)} algorithm distributes processing rates across all alive jobs proportional to their weights. Motwani et al.~\cite{DBLP:journals/tcs/MotwaniPT94} showed that the algorithm has a competitive ratio of~$2$ for jobs with uniform weights, and Kim and Chwa~\cite{DBLP:journals/ipl/KimC03a} proved the same competitive ratio for arbitrary weights. In both cases, this ratio is best possible. It is not difficult to see that WRR is monotone, since shrinking a job's processing requirement only decreases its completion time and thus gives all other jobs  more rate earlier, also reducing their completion time.

By~\Cref{coro:framework} we conclude with the following result.

\begin{theorem}\label{thm:single-machine}
  Preferential Time Sharing with the prediction-clairvoyant WSPT algorithm and the non-clairvoyant WRR algorithm has, for every~$\lambda \in (0,1)$,
  a competitive ratio of at most
  \[
  \min \left\{ \frac{1}{1 - \lambda}\left(1 + \frac{\errS}{\opt} \right), \frac{2}{\lambda} \right\}
  \]
  for non-clairvoyant scheduling with predictions~\abc{1}{pmtn}{\sum w_j C_j}.
\end{theorem}

\subsection{Identical parallel machines}\label{sec:identical-machines}

Consider non-clairvoyant scheduling of weighted jobs with release dates on~$m$ identical parallel machines, \abc{P}{r_j,pmtn}{\sum w_jC_j}. As prediction we assume a single permutation $\pred$ over all jobs, i.e., we do not require a machine assignment.

\paragraph{Prediction-clairvoyant algorithm.} Consider the {\em preemptive WSPT (P-WSPT)} algorithm that schedules, at any moment in time, the~$m$ available jobs with the highest priority in the predicted order~$\pred$. Assuming~$\pred$ is a perfect prediction and gives the true WSPT order, P-WSPT is known to be~$2$-competitive~\cite{DBLP:journals/orl/MegowS04}. Further notice that, for a fixed permutation prediction, smaller processing requirements will not increase the objective value of this algorithm. Thus, it is monotone. We show the following error-dependence.
\begin{lemma}\label{lemma:identical-machines-error}
  The prediction-clairvoyant P-WSPT algorithm is~$(\errS/m)$-error-dependent.
\end{lemma}
\begin{proof}
  Consider an instance~$J$ with jobs being indexed by~$\perfectpred$, a prediction~$\pred$, and the schedule obtained by the prediction-clairvoyant P-WSPT.
  After job~$j$ has been released, it is either being processed on a machine or it is delayed by another job.
	Let~$d(j',j)$ denote the total amount of job~$j'$ that delays the completion of~$j$. Note that~$d(j',j) \leq p_{j'}$.
  Such a delay can only occur if there are at least~$m$ alive jobs before~$j$ in~$\pred$, and these jobs will be distributed over all~$m$ machines.
  Since~$j$ has received $p_j$ units of processing by its completion time, we conclude
  \begin{align*}
      \alg(J, \pred) &\leq \sum_{j=1}^n w_j (r_j + p_j) + \frac{1}{m} \sum_{j=1}^n \sum_{j' = 1}^{j-1} \left( w_j \cdot d(j',j) + w_{j'} \cdot d(j,j') \right) \\
      &\leq  \opt(J) + \frac{1}{m} \sum_{j=1}^n \sum^{j-1}_{\substack{j' = 1\\\pred(j') < \pred(j)}} w_j p_{j'} + \frac{1}{m} \sum_{j=1}^n \sum^{j-1}_{\substack{j'=1\\ \pred(j') > \pred(j)}} w_{j'} p_j \\
      &=  \opt(J) + \frac{1}{m} \sum_{j=1}^n w_j \sum_{j' = 1}^{j-1} p_{j'} + \frac{1}{m} \sum_{j=1}^n \sum^{j-1}_{\substack{j'=1\\ \pred(j') > \pred(j)}} (w_{j'} p_j - w_j p_{j'}) \\
      &\leq 2 \cdot \opt(J) + \frac{1}{m} \cdot \errS(J,\pred).
  \end{align*}

  The second and third inequality hold due to two classical lower bounds on an optimal solution: Every job has to be processed by at least its %
  $p_j$ after its release in any solution. And~$\frac{1}{m} \sum_{j=1}^n w_j \sum_{j' = 1}^{j-1} p_{j'}$ equals the objective value of the WSPT schedule on a single machine with speed~$m$ without release dates, which is a known relaxation of our problem and therefore also a lower bound on~$\opt(J)$. Since we assumed that the jobs are indexed according to~$\perfectpred$, the sum of inversions is equal to~$\errS(J,\pred)$ by~\Cref{def:errS}.  %
\end{proof}

\paragraph{Non-clairvoyant algorithm.} Beaumont et el.~\cite{DBLP:conf/ipps/BeaumontBEM12} analyzed a natural extension of the WRR algorithm~\cite{DBLP:journals/ipl/KimC03a} %
to identical parallel machines and prove the same competitive ratio of~$2$ for non-clairvoyant \abc{P}{pmtn}{\sum w_jC_j}. Like WRR, their algorithm {\em Weighted Dynamic EQuipartition} (WDEQ) assigns processing rates to jobs  proportional to their weights making sure that no job receives a higher rate than executable on one machine simultaneously. 

When release dates are present, it is not hard to prove that WDEQ has a competitive ratio of at most~$3$. This result might be folkloric. %
To see it, consider the schedules~$S$ and~$S'$ of WDEQ with and without release dates for the same job set.
Let~$C_j$ resp.~$C'_j$ be the completion time of job~$j$ in~$S$ resp.~$S'$.
Notice that the total sum of rates job~$j$ receives in the interval~$[r_j, C_j]$ in~$S$ is not more than in the interval~$[0, C'_j]$ in~$S'$. This is because the total weight of other jobs running during~$[r_j, C_j]$ in~$S$ cannot be higher compared to the case when all jobs are released at the same time, which is the case in~$S'$. 
Thus,~$[r_j, C_j]$ is not longer than~$[0, C'_j]$, giving~$C_j \leq r_j + C'_j$.
The facts that~$\sum_{j} w_j r_j$ is a lower bound on the optimal objective value with release dates and~$\sum_{j} w_j C'_j$ is at most twice the optimal objective value without release dates~\cite{DBLP:conf/ipps/BeaumontBEM12} imply that WDEQ has a competitive ratio of at most~$3$ for~$P|r_j,pmtn|\sum w_j C_J$.

Note that WDEQ is monotone as shrinking a job only decreases its completion time and thus gives other jobs  more rate earlier, which also decreases their completion times.

\begin{lemma}%
  WDEQ is a monotone~$3$-competitive algorithm for the non-clairvoyant version of %
  \abc{P}{r_j,pmtn}{\sum w_j C_j}.
\end{lemma}

By~\Cref{coro:framework} we conclude with the following result.

\begin{theorem}\label{thm:identical-machine}
	Preferential Time Sharing with the prediction-clairvoyant P-WSPT algorithm and the non-clairvoyant WDEQ algorithm has, for every~$\lambda \in (0,1)$, a competitive ratio of 
  \[
  \min \left\{ \frac{1}{1 - \lambda}\left(2 + \frac{\errS}{m \cdot \opt} \right), \frac{3}{\lambda} \right\}
  \]
  for non-clairvoyant scheduling with predictions on $m$ identical parallel machines,~\abc{P}{r_j,pmtn}{\sum w_j C_j}.
\end{theorem}

\subsection{Unrelated machines}\label{sec:unrelated-machines}

We consider our most general non-clairvoyant scheduling problem, preemptive scheduling of weighted jobs on unrelated machines~\abc{R}{r_j,pmtn}{\sum w_j C_j}, with predictions. Recall that we are given a predicted permutation~$\pred_{i}$ for each machine~$i\in [m]$ including a predicted machine allocation~$m(\pred,j)$ for each job~$j$.

\paragraph{Prediction-clairvoyant algorithm}
The best known algorithm for clairvoyant scheduling~\abc{R}{r_j,pmtn}{\sum w_j C_j} by Bienkowski et al.~\cite{DBLP:conf/icalp/BienkowskiKL21} has a competitive ratio of~$3$. It uses a guess-and-double framework and processing times; it is unclear how to run it based on permutation predictions and how to track its error-dependence.

Other clairvoyant algorithms where proposed (for different problems)~\cite{DBLP:journals/mor/MegowUV06,DBLP:conf/soda/AnandGK12,DBLP:journals/mor/GuptaMUX20,DBLP:phd/dnb/Jager21} that use a greedy strategy for assigning jobs to machines in the following way. Assuming a fixed single-machine rule~$\Pi$, they assign a newly arriving job to the machine where it causes the (approximately) minimum increase in the objective value, assuming that jobs on each machine are scheduled according to~$\Pi$. We refer to such algorithm as~{\em MinIncrease~$\Pi$}.

While these algorithms are similar in flavor, none of the existing results proven in the literature seems to directly match our purpose w.r.t.\ the precise scheduling model and the possibility for proving an error-sensitivity.

Most promising seems a result for minimizing the total weighted flow time on unrelated machines, where the flow time of a job~$j$ is defined as~$C_j-r_j$. Anand et al.~\cite{DBLP:conf/soda/AnandGK12} use the \emph{Weighted Shortest Remaining Processing Time first (WSRPT)} rule as single-machine algorithm~$\Pi$, which schedules, at any time~$t$, an available job with largest residual density~$w_j / p_{j}(t)$. 
For the (simpler) objective of minimizing the weighted completion time, WSRPT is known to be~$2$-competitive on a single machine with release dates~\cite{Megow07-Diss}. A straightforward adaption of the analysis in~\cite{DBLP:conf/soda/AnandGK12} shows that MinIncrease WSRPT is~$8$-competitive for our clairvoyant problem. A more careful analysis even proves a competitive ratio of at most~$4$~\cite{DBLP:phd/dnb/Jager21}.
However, it is unclear how to turn this algorithm into a prediction-clairvoyant algorithm in our setting. While the machine assignment is given, we do not have information about (remaining) processing times to apply WSRPT. Further, it is unclear how to obtain an error-dependency for our permutation prediction model, as the order of the jobs given by WSRPT changes when jobs are processed.

Nevertheless, we take inspiration from the analysis, replace \mbox{WSRPT} by preemptive WSPT and adopt the MinIncrease P-WSPT algorithm for our framework. %
We first prove that the clairvoyant MinIncrease P-WSPT algorithm is at most~$5.8284$-competitive using a dual-fitting analysis borrowing different ideas from~\cite{DBLP:phd/dnb/Jager21,DBLP:conf/soda/AnandGK12,DBLP:journals/mor/GuptaMUX20}. Without release dates, our algorithm is essentially the same as the algorithms in~\cite{DBLP:phd/dnb/Jager21,DBLP:conf/soda/AnandGK12,DBLP:journals/mor/GuptaMUX20} and a 
lower bound of $4$ is known~\cite{DBLP:journals/mor/GuptaMUX20,correaQ12unrelatedLB}.
We also prove an error-dependent competitive ratio for its prediction-clairvoyant version with release dates.

\begin{restatable}{theorem}{theoremGreedyWSPT}\label{theorem:unrelated-clairvoyant}
  The MinIncrease P-WSPT algorithm has a competitive ratio of at most~$3 + 2 \sqrt{2} \approx 5.8284$ for clairvoyant scheduling on unrelated machines,~\abc{R}{r_j, pmtn}{\sum w_j C_j}.
\end{restatable}

In the following we denote the MinIncrease P-WSPT algorithm by~$\A$.
Fix an instance~$J$ and let~$\speed > 1$ be a real number %
that we will fix later. 
We assume w.l.o.g.\ by scaling the instance that all processing requirements and release dates in~$J$ are integer multiples of $\speed$. 

Let~$M_i(j)$ be the set of available jobs that are assigned to machine~$i$ at time~$r_j$, excluding job~$j$. As this definition is ambiguous if there are two jobs~$j$ and~$j'$ with~$r_{j} = r_{j'}$ being assigned to~$i$, we assume that we assign them in the order of their index. By defining~$\dense_{ij} = w_j / p_{ij}$, the increase of the objective value of~$\A$ due to assigning job~$j$ to machine~$i$ at time~$r_j$ equals
\[
	Q_{ij} = w_{j} \Bigg( r_j + p_{ij} + \sum_{\substack{j' \in M_i(j) \\ \dense_{ij'} \geq \dense_{ij}}} p_{ij'}(r_j) \Bigg) + p_{ij} \sum_{\substack{j' \in M_i(j) \\ \dense_{ij'} < \dense_{ij}}} w_{j'}.
\]

Then, algorithm \A assigns job~$j$ to machine~$g(j) = \argmin_i Q_{ij}$. 
The following linear program is a relaxation of our scheduling problem~\cite{DBLP:conf/soda/AnandGK12,DBLP:journals/mor/GuptaMUX20,DBLP:phd/dnb/Jager21}. The variable~$x_{ijt}$ denotes the fractional assignment of job~$j$ to machine~$i$ at time~$t$.
  \begin{alignat}{3}
    \text{min} \quad &\sum_{i,j,t} w_{j} \cdot \left( \frac{x_{ijt}}{2} + \frac{x_{ijt}}{p_{ij}} \cdot \left( t + \frac{1}{2} \right) \right) \tag{LP}\label{lp} \\ 
    & \sum_{i,t \geq r_j} \frac{x_{ijt}}{p_{ij}} \geq 1   &&\forall j \notag \\
    & \sum_{j} x_{ijt} \leq 1   &&\forall i, t \notag \\
    & x_{ijt} \geq 0 &&\forall i,j,t \notag \\
    & x_{ijt} = 0 &&\forall i,j,t < r_j \notag 
\end{alignat}

The dual of~\eqref{lp} is equal to the following linear program with variables~$\dualVa_j$ and~$\dualVb_{{it}}$.
\begin{alignat}{3}
    \text{max} \quad &\sum_{j} \dualVa_j - \sum_{i,t} \dualVb_{it} \tag{DLP}\label{dual} \\ 
    &\frac{\dualVa_j}{p_{ij}} \leq \dualVb_{it} + w_{j} \cdot \left(\frac{t + 1/2}{p_{ij}} + \frac{1}{2} \right) \qquad &\forall i,j,t \geq r_j \label{constr:dual} \\ 
    &\dualVa_j, \dualVb_{it}  \geq 0 \qquad &\forall i,j,t \notag
\end{alignat}

We define a solution of~\eqref{dual} for instance~$J$ which depends on the schedule produced by algorithm \A for~$J$. Let~$U_{i}(t) = \{ j \in J \mid g(j) = i \land t < C_j \}$, where~$C_j$ denotes the completion time of job~$j$ in the schedule of~\A for instance~$J$. 
Note that~$U_{i}(t)$ includes unreleased jobs. %
Consider the following assignment: %
\begin{itemize}[\quad ]
  \item $\dualSa_j = Q_{g(j)j}$ for every job~$j$ and 
  \item $\dualSb_{it} = \sum_{j \in U_{i}(\speed \cdot t)} w_{j}$ for every machine~$i$ and time~$t$.
\end{itemize}

We first show that the objective value of~\eqref{dual} for the solution~$(\dualSa_j,\dualSb_{it})$ is close to the objective value of %
\A w.r.t.~$\speed$.

\begin{lemma}\label{lemma:dual-objective}
   $\sum_{j} \dualSa_j - \sum_{i,t} \dualSb_{it} = \left( 1 - \frac{1}{\speed} \right) \cdot \A(J)$.
\end{lemma}

\begin{proof}
    The definition of~$Q_{g(j)j}$ implies~$\sum_j \dualSa_j = \sum_j Q_{g(j)j} = \A(J)$. 
  Since we assumed that all release dates and processing times in $J$ are integer multiples of $\speed$, all preemptions occur at integer multiples of $\speed$ and therefore also all job completions. Thus, $\sum_{t} \sum_{j \in U_i(\speed \cdot t)} w_j = \frac{1}{\speed} \sum_{t} \sum_{j \in U_i(t)} w_j$ for every machine~$i$, and %
    \[
      \sum_{i,t} \dualSb_{it} = \sum_{i,t} \sum_{j \in U_{i}(\speed \cdot t)} w_{j} = \frac{1}{\speed} \sum_{i,t} \sum_{j \in U_{i}(t)} w_{j} = \frac{1}{\speed} \cdot \A(J),
    \]
    which implies the desired equality.
\end{proof}

Second, we show that scaling the defined variables makes them feasible for~\eqref{dual}.

\begin{lemma}\label{lemma:dual-feasible}
   Assigning $\dualVa_j = \dualSa_j/(\speed + 1)$ and $\dualVb_{it} = \dualSb_{it}/(\speed + 1)$ gives a feasible solution for~\eqref{dual}.
\end{lemma}

\begin{proof}
  Since our defined variables are non-negative by definition, it suffices to show that this assignment satisfies~\eqref{constr:dual}.
  Fix a job~$j$, a machine~$i$ and a time~$t \geq r_j$.
  We assume that no new job arrives after~$j$, since such a job may only increase~$\dualSb_{it}$ while~$\dualSa_j$ stays unchanged. Let~$j_1,\ldots,j_z$ be the jobs of~$M_i(j)$ indexed in WSPT order by densities $\mu_{ij}=w_j/p_{ij}$.
  Defining
  \begin{itemize}[\quad]
      \item $H = \{j' \in M_i(j): \dense_{ij'} \geq \dense_{ij} \} = \{j_1,\ldots,j_r\}$ and
      \item $L = \{j' \in M_i(j): \dense_{ij'} < \dense_{ij} \} = \{j_{r+1},\ldots,j_z\}$,
  \end{itemize}
  and using~$\dualSa_j = Q_{g(j)j} \leq Q_{ij}$ and $\speed + 1 > 2$ yields
  \[
    \frac{\dualVa_j}{p_{ij}} 
    = \frac{\dualSa_j}{(\speed + 1)p_{ij}} \leq \frac{\dense_{ij}}{\speed + 1} \left(r_j + \sum_{j' \in H} p_{ij'}(r_j) \right) + \frac{w_j}{2} + \sum_{j' \in L} \frac{w_{j'}}{\speed + 1}.
  \]
  
  Thus, asserting~\eqref{constr:dual} reduces to proving
  \begin{equation}
    \frac{\dense_{ij}}{\speed + 1} \left(r_j + \sum_{j' \in H} p_{ij'}(r_j) \right) + \sum_{j' \in L} \frac{w_{j'}}{\speed + 1} \leq   \dense_{ij} t + \dualVb_{it}. \label{eq:to-prove}
  \end{equation}

  Observe that the total processing time of all jobs in~$M_i(j)$ that are completed before time~$\speed \cdot t$ is at most~$\speed \cdot t$. Further,~$r_j + \speed \cdot t \leq (\speed + 1) t$.
  Now consider the case that machine~$i$ processes a job~$j_k$ at time~$\speed \cdot t$.
  If~$j_k \in H$, using~$\dense_{ij} \leq \frac{w_{j_\ell}}{p_{ij_\ell}} \leq \frac{w_{j_\ell}}{p_{ij_\ell}(r_j)}$ for all~$j_\ell \in H$ gives
  \begin{align*}
    &\frac{\dense_{ij}}{\speed + 1} \left(r_j + \sum_{\ell=1}^{k-1} p_{ij_\ell}(r_j) \right)  + \frac{\dense_{ij}}{\speed + 1} \sum_{\ell=k}^{r} p_{ij_\ell}(r_j)  + \sum_{j' \in L} \frac{w_{j'}}{\speed + 1} \\
    & \leq \dense_{ij} t + \frac{1}{\speed + 1}\sum_{\ell=k}^{r} w_{j_\ell}  + \sum_{j' \in L} \frac{w_{j'}}{\speed + 1} \leq \dense_{ij} t  + \frac{\dualSb_{it}}{\speed + 1} = \dense_{ij} t + \dualVb_{it}.
  \end{align*}

  The last inequality holds since all jobs in~$M_i(j)$ that are processed after job~$j_{k-1}$ are unfinished at time~$\speed \cdot t$ and assigned to~$i$ in~$\A's$ schedule, hence part of~$U_{i}(\speed \cdot t)$.
  If~$j_k \in L$, using~$w_{j_\ell} < \dense_{ij} \cdot p_{ij_\ell}$ for all~$j_\ell \in L$ implies
  \begin{align*}
    & \frac{\dense_{ij}}{\speed + 1} \left(r_j + \sum_{\ell=1}^{r} p_{ij_\ell}(r_j) \right)  + \frac{1}{\speed + 1} \sum_{\ell=r+1}^{k-1} w_{j_\ell} + \frac{1}{\speed + 1} \sum_{\ell=k}^{z} w_{j_\ell}  \\
    & \leq \frac{\dense_{ij}}{\speed + 1} \left(r_j + \sum_{\ell=1}^{r} p_{ij_\ell}(r_j) + \sum_{\ell=r+1}^{k-1} p_{ij_\ell} \right) + \frac{1}{\speed + 1} \sum_{\ell=k}^{z} w_{j_\ell}  \\ 
    & \leq \dense_{ij} t + \frac{1}{\speed + 1} \sum_{\ell=k}^{z} w_{j_\ell} \leq \dense_{ij} t + \frac{\dualSb_{it}}{\speed + 1} = \dense_{ij} t + \dualVb_{it}.
  \end{align*}
  
  If no job is running at time~$\speed \cdot t$, we conclude that all jobs in~$M_i(j)$ must already be completed, because algorithm \A does not idle unnecessarily, and we assumed that no job is released after~$j$. By using~$w_{j_\ell} < \dense_{ij} \cdot p_{ij_\ell}$ for all~$j_\ell \in L$ we assert~\eqref{eq:to-prove} for this final case
  \[
    \frac{\dense_{ij}}{\speed + 1} \left(r_j + \sum_{\ell=1}^{r} p_{ij_\ell}(r_j) \right) + \sum_{\ell=r+1}^{z} w_{j_\ell}  
    \leq \frac{\dense_{ij}}{\speed + 1} \left(r_j + \sum_{\ell=1}^{r} p_{ij_\ell}(r_j) + \sum_{\ell=r+1}^{z} p_{ij_\ell} \right)
     \leq \dense_{ij} t. \qedhere
  \]
\end{proof}

\begin{proof}[Proof of~\Cref{theorem:unrelated-clairvoyant}]
  Weak duality and \Cref{lemma:dual-feasible} imply that the objective value of~\eqref{dual} for the assigned variables is a lower bound on the optimal objective value. \Cref{lemma:dual-objective} gives 
   \begin{equation*} 
    \opt(J) \geq \sum_{j \in J} \dualVa_j - \sum_{i,t} \dualVb_{it} = \sum_{j \in J} \frac{\dualSa_j}{\speed + 1} - \sum_{i,t} \frac{\dualSb_{it}}{\speed + 1} = \frac{1}{\speed + 1} \left( \sum_{j \in J} \dualSa_j - \sum_{i,t} \dualSb_{it} \right) 
    = \left( \frac{1 - 1/\speed}{\speed + 1} \right) \cdot \A(J). 
  \end{equation*}

  We conclude that algorithm \A has a competitive ratio of at most~$3 + 2 \sqrt{2} \approx 5.8284$ for the optimal choice~$\speed = 1 + \sqrt{2}$.
\end{proof}

We now consider the prediction-clairvoyant version of the MinIncrease P-WSPT algorithm. It  assigns the jobs to machines according to the predicted assignment~${\{\pred_i\}}_{i\in[m]}$. At any time and for every machine~$i$ it schedules the available job with highest priority according to~$\pred_i$.
This algorithm is monotone, because shrinking a job's processing requirements does not affect~${\{\pred_i\}}_{i\in[m]}$ and thus may only decrease the completion time of jobs that are scheduled after this job on the assigned machine.

\begin{lemma}
  The prediction-clairvoyant MinIncrease P-WSPT algorithm is~$\errR$-error-dependent.
\end{lemma}

\begin{proof}
  Consider %
  job set~$J$.
  Scheduling a job~$j$ according to a prediction~$\pred$ contributes a value equal to~$ W_j(J,\pred)$ to the objective of our algorithm~$\alg(J, \pred)$. Thus, $\alg(J, \pred) = \sum_j W_j(J,\pred)$.  
  Since the machine assignment of the clairvoyant MinIncrease P-WSPT algorithm~\A can be encoded into a prediction that orders the jobs by WSPT on every machine, the cost of following~$\sigma$ is a lower bound on the objective value of~\A, i.e.~$\sum_j W_j(J,\perfectpred) \leq \A(J)$, or $-\A(J) \leq -\sum_j W_j(J,\perfectpred)$.
  We conclude using \Cref{theorem:unrelated-clairvoyant} that our algorithm is~$\errR$-error-dependent, since
  \begin{align*}
    \alg(J, \pred) &= \A(J) - \A(J) + \sum_{j \in J} W_j(J,\pred) 
    \leq \A(J) + \sum_j W_j(J,\pred) - \sum_{j \in J} W_j(J,\perfectpred) \\
    &= \A(J) + \errR(J,\pred) \leq 5.8284 \cdot \opt(J) + \errR(J,\pred). \qedhere
  \end{align*}
\end{proof}

\paragraph{Non-clairvoyant algorithm} Im, Kulkarni and Munagala~\cite{DBLP:journals/jacm/ImKM18} show that the Proportional Fairness (PF) algorithm is $128$-competitive. %
	(They actually state a smaller competitive ratio of $64$ in~\cite[Theorem 1.2]{DBLP:journals/jacm/ImKM18} but there is missing a factor of $2$ when applying Lemma 3.2 and Corollary 3.5.)  %
A similar argumentation as for WRR and WDEQ shows that this algorithm is monotone.

\begin{lemma}[\cite{DBLP:journals/jacm/ImKM18}]\label{lemma:unrelated-non-clairvoyant}
  The Proportional Fairness algorithm is~$128$-competitive for non-clairvoyant scheduling on unrelated machines,~\abc{R}{r_j, pmtn}{\sum_j w_j C_j}.
\end{lemma}

By~\Cref{coro:framework} we conclude with the following result.

\begin{theorem}\label{thm:main-unrelated}
  Preferential Time Sharing with the prediction-clairvoyant MinIncrease P-WSPT algorithm and the non-clairvoyant Proportional Fairness algorithm has, for every~$\lambda \in (0,1)$, 
  a competitive ratio of at most
  \[
  \min \left\{ \frac{1}{1 - \lambda}\left(5.8284 + \frac{\errR}{\opt} \right), \frac{128}{\lambda} \right\}
  \]
  for non-clairvoyant scheduling with predictions on unrelated machines,~\abc{R}{r_j, pmtn}{\sum_j w_j C_j}.
\end{theorem}

\section{Learnability of permutations}\label{sec:learnability}
We show that permutation predictions for identical machines are PAC-learnable in the agnostic sense w.r.t.~$\errS$. %

\begin{theorem}
    For any~$\epsilon, \delta \in (0,1)$ and any distribution~$\D$ over the instances of length~$n$, there exists a learning algorithm which, given an i.i.d.\ sample of~$\D$ of size~$z \in \bigO \left(\frac{1}{\epsilon^2}\cdot {(n \log n - \log \delta)n^2}\right)$, returns in polynomial time depending on~$n$ and~$z$ a prediction~$\pred_p \in \Hyp$ from the set of all possible permutations of the set~$\{1,\ldots,n\}$, such that with probability of at least~$(1-\delta)$ it holds~$\E_{J \sim \D}[\errS(J, \pred_p)] \leq \E_{J \sim \D}[\errS(J,\perfectpred)] + \epsilon$, where $\errS(J,\pred)$ denotes the error of $\pred$ for instance~$J$, and $\perfectpred = \argmin_{\pred \in \Hyp} \E_{J \sim \D}[\errS(J,\pred)]$.
\end{theorem}

\begin{proof}
    Let~$\epsilon, \delta \in (0,1)$.
    We prove that we can use the classic Empirical Risk Minimization (ERM) learning method to find such a prediction. Let~$\cS = \{J_1,\ldots,J_z\}$ be a set of i.i.d.\ samples from~$\D$. The ERM method then determines the prediction that minimizes the empirical error~$\errS_\S(\pred) = \frac{1}{z} \sum_{s=1}^z \errS(J_s,\pred)$. Since there are~$n!$ possible permutations of the set~$\{1,\ldots,n\}$, we conclude that~$\Hyp$ is finite, and we can assume by scaling processing requirements and weights to~$[0,1]$ that our error function is bounded by~$n$. Classic results, see~e.g.~\cite{shalevB14ML}, imply for this case that~$\Hyp$ is agnostically PAC learnable using the ERM method with sample complexity
    \[
		z \leq  \left\lceil \frac{2 \log(2 |\Hyp| / \delta)n^2}{\epsilon^2}\right\rceil \in \bigO \left( \frac{(n \log n - \log \delta)n^2}{\epsilon^2} \right),
	\]
   which is polynomial in the number of jobs, $n$, as $\log{n!} \in \bigO(n \log n)$.
    
    It remains to prove that the ERM algorithm can be implemented efficiently in our setting, that is, given a sample set of size $z$, determine in time polynomial in $n$, a prediction that minimizes the empirical error. Rewriting the empirical error gives
    \[
        \errS_\S(\pred) = \frac{1}{z} \sum_{s=1}^z \errS(J_s,\pred) = \frac{1}{z} \sum_{s=1}^z \sum_{j=1}^n \left( W_j(J_s, \pred) - W_j(J_s, \perfectpred) \right).
    \]

    Since the values~$W_j(J_s, \perfectpred)$ are independent of~$\pred$, it suffices to find a prediction~$\pred$ that minimizes~$\frac{1}{z} \sum_{s=1}^z \sum_{j=1}^n W_j(J_s, \pred)$. 
    For the special error~$\errS$, by denoting for a job~$j \in J_s$ its weight by $w^{(s)}_j$ and its processing requirement by $p^{(s)}_j$, this is equal to

    \begin{equation*}
      \frac{1}{z} \sum_{s=1}^z  \sum_{j=1}^n W_j(J_s, \pred) = \frac{1}{z} \sum_{s=1}^z \sum_{j=1}^n w^{(s)}_{\pred(j)} \sum_{\ell = 1}^j p^{(s)}_{\pred(\ell)} = 
        \sum_{j=1}^n \left( \frac{1}{z} \sum_{s=1}^z w^{(s)}_{\pred(j)} \right) \sum_{\ell = 1}^j \left( \frac{1}{z} \sum_{s=1}^z p^{(s)}_{\pred(\ell)} \right).
    \end{equation*}

    By defining the average weight~$\bar{w}_{\pred(j)} = \frac{1}{z} \sum_{s=1}^z w^{(s)}_{\pred(j)}$ and average processing requirement~$\bar{p}_{\pred(j)} = \frac{1}{z} \sum_{s=1}^z p^{(s)}_{\pred(j)}$ over $\S$ for all~$j \in [n]$, this is equal to minimizing
    \[
      \sum_{j=1}^n \bar{w}_{\pred(j)} \sum_{\ell = 1}^j \bar{p}_{\pred(\ell)}.
    \]

    Consider the {\em average} instance of $\S$, i.e.\ the scheduling instance of $n$ jobs with weights~${\{\bar{w}_j\}}_{j \in [n]}$ and processing requirements~${\{\bar{p}_j\}}_{j \in [n]}$. 
    Since the above expression is equal to the objective value of this instance when scheduling jobs in order~$\pred(1),\ldots,\pred(n)$, we can minimize it by ordering the jobs according to WSPT in polynomial time in~$z$ and~$n$~\cite{smith1956various}.
    \end{proof}

   The space of permutation predictions with predicted machine assignments~${\{\pred_i\}}_{i\in[m]}$ is also finite and we can use similar arguments to prove that they are agnostically PAC-learnable with respect to~$\errR$ with bounded sample complexity.
    This implies that ERM minimizes the empirical error. However, it is not clear how to achieve this with polynomial running time in~$n$, $m$ and the number of samples $z$. 
    Yet one can approximately minimize the empirical error by computing an approximately perfect prediction using the MinIncrease P-WSPT algorithm on the average instance of~$\S$.

\section{Experiments}\label{sec:experiments}

In empirical experiments\footnote{The code is available on GitHub:~\url{https://github.com/mountlex/nonclairvoyant-scheduling/tree/arxiv-v2}} we demonstrate the practicability of our approach in comparison to the previously proposed learning-augmented algorithms by Im~et~al.~\cite{DBLP:conf/spaa/Im0QP21} and Wei and Zhang~\cite{DBLP:conf/nips/WeiZ20}. 
These algorithms consider the {\em single-machine} problem without weights and release dates, \abc{1}{pmtn}{\sum C_j}. %
Notice that in this setting the Preferential Time Sharing~(PTS) algorithm and the Preferential-Round-Robin~(PRR) algorithm of Kumar et al.~\cite{DBLP:conf/nips/PurohitSK18} are equivalent. The only difference is the theoretically different prediction model. However, since all previous algorithms use the length prediction model, we compute permutation predictions based on predicted processing times.
The results for the single machine setting are given in \Cref{sec:exp-single}. 
We further give experimental results for PTS for scheduling weighted jobs  with release dates on parallel identical machines in~\Cref{sec:exp-multiple}. But first we describe the instance generation and experiment setups.

\paragraph{Dataset}
We generate synthetic instances. Each instance is composed of 1000 jobs. We choose this size as a compromise between computational effort and giving the algorithms enough jobs to work properly. The processing requirements for the jobs are individually sampled from a Pareto  distribution with scale~$1$ and shape~$1.1$. This distribution was used in the seminal work on learning-augmented scheduling~\cite{DBLP:conf/nips/PurohitSK18} 
and is (similar to the related Zipf distribution) generally considered to model scheduling applications very well~\cite{DBLP:conf/sigmetrics/BansalH01,DBLP:journals/ton/CrovellaB97,DBLP:journals/tocs/Harchol-BalterD97,DBLP:conf/stacs/ImMP15,DBLP:journals/glottometrics/AdamicH02,DBLP:books/daglib/0025903,DBLP:books/daglib/0034524}. Intuitively, it gives many tiny jobs and few very large jobs. We also performed our experiments with processing requirements sampled from an exponential distribution with mean~1 as well as from a Weibull distribution with scale~2 and shape~$0.5$, which were used in~\cite{DBLP:conf/acda/Mitzenmacher21,DBLP:conf/innovations/Mitzenmacher20}.

\paragraph{Types of experiments}
We perform two types of experiments.
In \emph{sensitivity experiments}, we generate length predictions by adding Gaussian noise to the processing requirements with an increasing standard deviation~$\omega$ for a fixed instance. 
This type of experiment was also performed by Kumar et al.~\cite{DBLP:conf/nips/PurohitSK18} to evaluate PRR as well as in other works on learning-augmented algorithms~\cite{DBLP:conf/icml/LykourisV18,DBLP:conf/icml/AntoniadisCE0S20,AntoniadisCEPS21,DBLP:conf/innovations/LindermayrMS22}.

In \emph{online learning experiments} we first fix a synthetic instance, called base instance. Then, we consider~$10$ subsequent rounds, where in every round~$t$ an instance~$J_t$ arrives, which is generated by adding independently sampled Gaussian noise to the base instance. 
To calculate this noise we use scaled standard deviations parameterized by a factor~$\gamma \geq 0$. That is, we compute noise for true processing requirement~$p$ with a standard deviation equal to~$\gamma \cdot \sqrt{p}$. We feel that this is more realistic for this type of experiment than only using a fixed standard deviation for all jobs, as small jobs may vary less than large jobs over time.
We then compute a prediction for round~$t$ using the ERM algorithm on the set of previous instances~$\{J_0,\ldots,J_{t-1}\}$, as these are in round~$t$ known to the algorithms. As length prediction for round 0 we use an independently sampled random instance.
This type of experiment was also performed in~\cite{DinitzILMV21} to demonstrate the speedup of predictions for the bipartite matching problem.

\subsection{Experiments for a single machine}\label{sec:exp-single}

\paragraph{Algorithms} We present implementation details of the considered algorithms.
As online benchmark (without predictions), we use the best-possible non-clairvoyant algorithm Round-Robin~(RR)~\cite{DBLP:journals/tcs/MotwaniPT94}.

{\em TwoStage} (algorithm by Wei and Zhang~\cite{DBLP:conf/nips/WeiZ20}) executes RR until a certain time point depending on the predicted processing requirements and the confidence parameter~$\lambda \in [0,1]$. Then, it schedules the jobs in non-decreasing order of their predicted processing requirements. If at any time a job finishes before or after their predicted length, it finishes the remaining instance with RR. This algorithm achieves for instances with at most two jobs a consistency-robustness tradeoff that matches a lower bound shown in~\cite{DBLP:conf/nips/WeiZ20}.

\emph{MultiStage} (algorithm by Im et al.~\cite{DBLP:conf/spaa/Im0QP21}), works in %
phases and decides whether to follow the prediction or to execute RR by tracking the quality of the prediction. This is done by processing and computing the error of small random samples, which is then projected to the whole set of remaining jobs. We implemented a basic variant of this algorithm, which is~$\bigO(1)$-robust and~$(1+\epsilon)$-consistent for any~$\epsilon > 0$ with high probability under some assumptions (Corollary~34 in~\cite{DBLP:conf/spaa/Im0QP21}).
Our implementation uses base two for unspecified logarithms. A consequence of this choice is that if~$\epsilon < 0.215$, MultiStage executes solely RR on our instances. Therefore, we performed the experiments with rather large~$\epsilon = 0.25$ and~$\epsilon = 10.0$.
The authors of~\cite{DBLP:conf/spaa/Im0QP21} also give a modification of this algorithm which achieves bounds in expectation. We omitted the implementation of this variant as it requires a further parallel execution of RR which makes the calculation of precise completion times very difficult.

\begin{figure}
  \begin{subfigure}[t]{0.49\textwidth}
  \includegraphics[width=\linewidth]{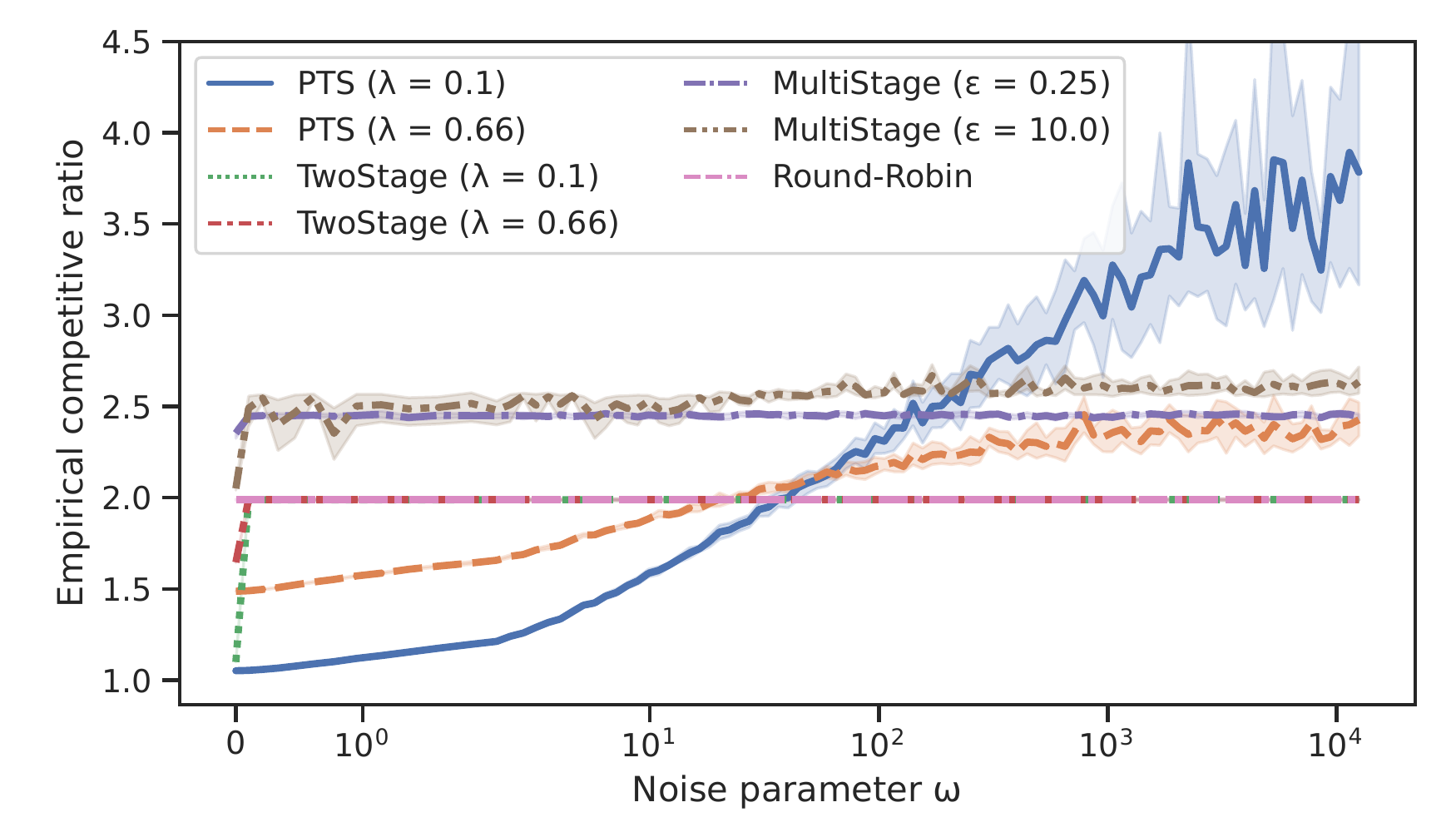}
  \caption{Sensitivity experiment.}\label{fig:exp1}
\end{subfigure}
\hfill
\begin{subfigure}[t]{0.49\textwidth}
  \centering
  \includegraphics[width=\linewidth]{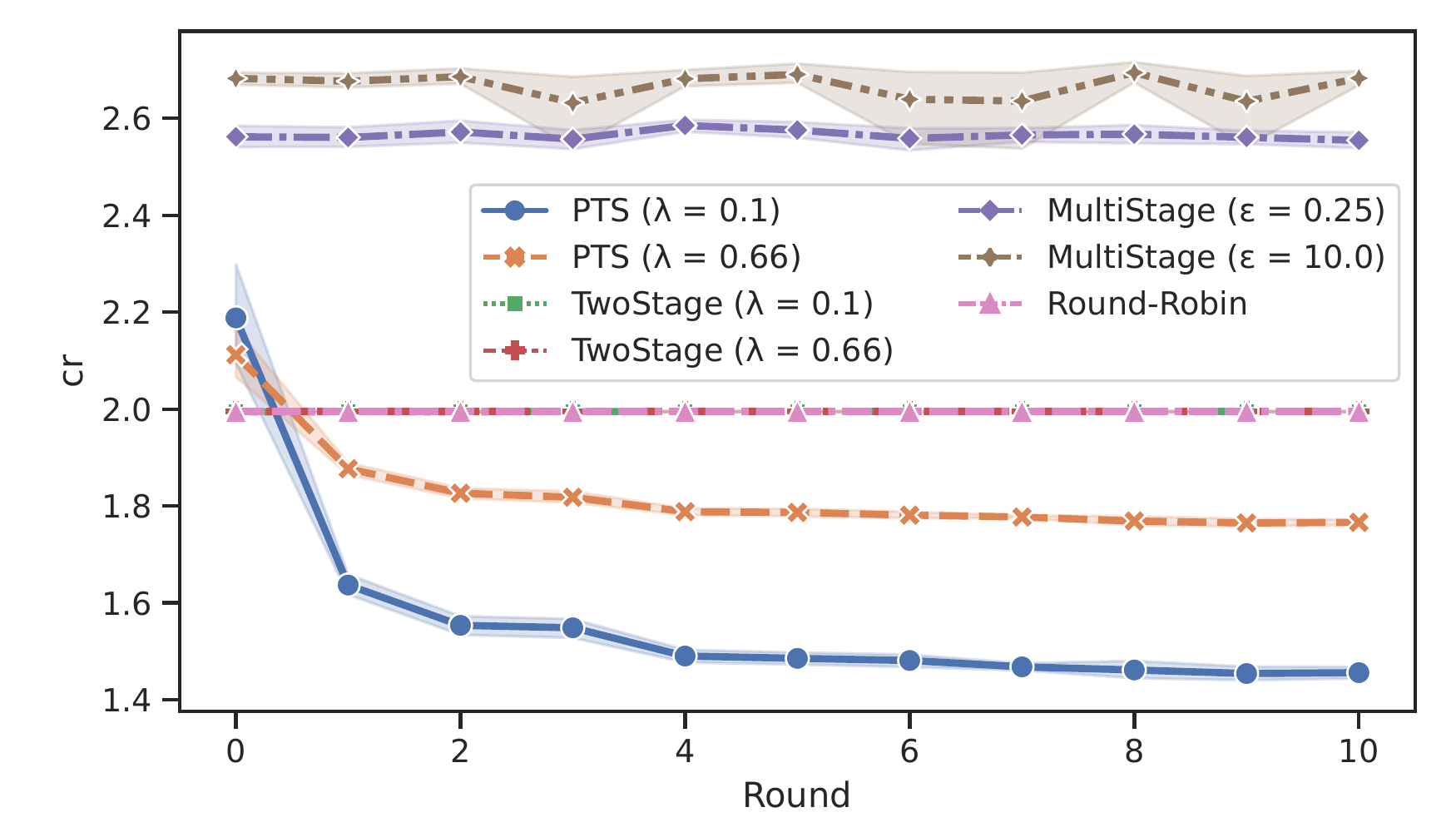}
  \caption{Online learning experiment,~$\gamma = 10$. Note that the %
  plots of TwoStage coincide with the plot of Round-Robin.}\label{fig:exp2}
\end{subfigure}
\caption{Single machine experiments}
\end{figure}

\paragraph{Results}
For every parameter setting we perform~$10$ runs and measure the performance of the algorithms for this setting in terms of \emph{empirical competitive ratio}. That is the average objective value of an algorithm over all runs divided by the optimal objective value for the instance.
We further report error bars that denote the~95\% confidence interval of the runs.

In the following we discuss results for Pareto-distributed processing requirements. For the other considered distributions, we observed very similar results, where in the online learning experiment we use varied noise parameters due to different job characteristics.

We first discuss results of the sensitivity experiment, which are visualized in~\Cref{fig:exp1}. For the consistency case~($\omega = 0$) the algorithms achieve their best performance, as expected. However, even for very small noise ($\omega = 0.1$), we observe that TwoStage and MultiStage experience drastic performance losses compared to having access to precise predictions. This behavior is explainable by the design of the algorithms, 
which switch their execution to the robust fallback procedure RR when detecting incorrect predictions. While TwoStage stays in this mode until the instance completes, MultiStage still estimates medians and errors, incurring an additional overhead.
While the performance of PTS smoothly degrades for larger noise depending on~$\lambda$, it still outperforms RR until~$\omega \approx 20$. For very large noise, the performance of TwoStage and MultiStage stays unchanged, while PTS with~$\lambda=0.1$ still grows.
For larger values of~$\lambda$, e.g.~$\lambda=0.66$ as in the figure, PTS shows a constantly superior performance than MultiStage and, w.r.t.\ TwoStage, a smoother performance with substantially better consistency and only slightly larger robustness. 

In the online learning experiment (\Cref{fig:exp2}), TwoStage and MultiStage do not improve their performance over RR by using predictions. We suspect that this is again due to the fact that the prediction is still too erroneous over the first ten rounds to activate their trustful subroutines. We performed these experiments also with 100 rounds, but did not observe a significant difference. While in round 0 without any prediction PTS performs slightly worse than the other algorithms, it improves over RR already after seeing one sample. This shows that in our setup one sample is enough to approximately distinguish small jobs from large jobs, and this classification is enough to prevent large jobs from delaying the completion of many small jobs. This also demonstrates that permutation predictions capture the relevant information of practical~instances.

\subsection{Experiments for multiple machines}\label{sec:exp-multiple}

We generate 10 synthetic instances with~1000 jobs each. Processing requirements are again sampled from a Pareto-distribution with shape~1.1 and scale~1, weights and release dates from a Pareto-distribution with shape~2 and scale~1. We implement PTS according to~\Cref{thm:identical-machine} and compare it to the non-clairvoyant WDEQ algorithm~\cite{DBLP:conf/ipps/BeaumontBEM12}. To compute empirical competitive ratios and error bars, we use the objective value of the clairvoyant, 2-competitive P-WSPT algorithm~\cite{DBLP:journals/orl/MegowS04} as baseline.
The results of the sensitivity experiment for five machines (\Cref{fig:exp3}) show that for small noise PTS outperforms WDEQ. For growing noise the performance of PTS slowly degrades, but still improves upon WDEQ until $\omega \approx 35$. For large values of $\omega$, the empirical competitive ratio of PTS with~$\lambda = 0.1$ continues growing, while for~$\lambda = 0.5$ and~$\lambda=0.8$ the ratios quickly converge to their robustness bounds.

\begin{figure}
  \centering
  \includegraphics[width=.6\linewidth]{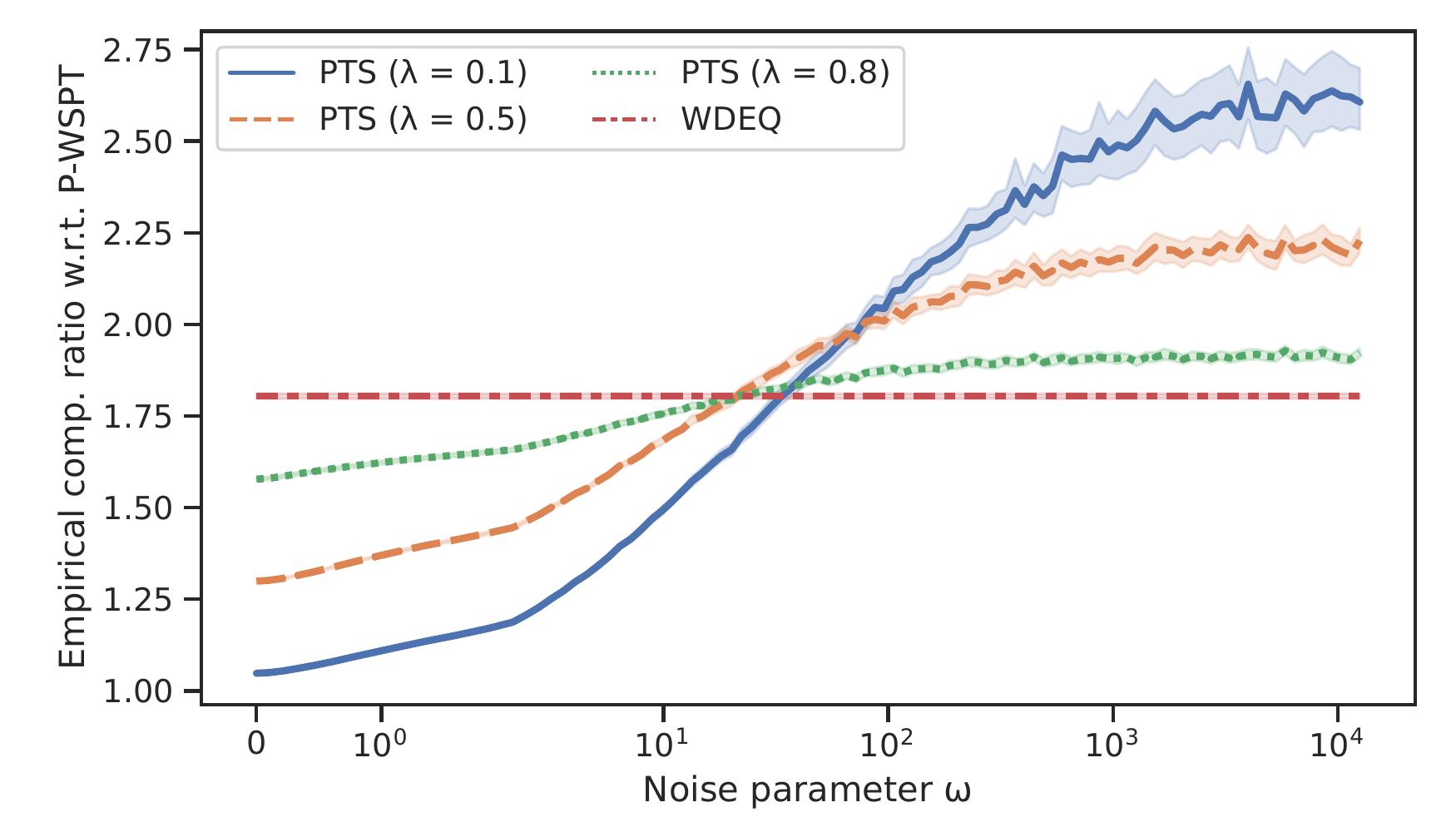}
  \caption{Sensitivity experiment for five identical parallel machines.}\label{fig:exp3}
\end{figure}

\section{Conclusion}

In this paper we proposed a new %
compact prediction model and error measure which fulfill desired properties in theory and practice. %
We revisited a learning-augmented time sharing framework, generalized it, and derived the first results for more complex scheduling problems with weights, release dates and multiple machines.

It would be interesting whether better guarantees are possible by exploiting the fact that processing at a slower rate makes jobs ``earlier'' available, or by exploiting communication between combined algorithms, or by more adaptive algorithms.

\bibliography{literature}

\end{document}